\documentclass[runningheads]{llncs}
\usepackage{amsmath,amssymb,amsfonts,mathrsfs}
\usepackage{stmaryrd} %%double brackets
\usepackage{graphicx}
\usepackage{booktabs}
\usepackage{dsfont}
\usepackage{hyperref}
\usepackage{array}
\usepackage{color}
\def\a0{i_0}

\newcommand{\eqdef}{\triangleq}
\renewcommand{\leq}{\leqslant}

\renewcommand{\geq}{\geqslant}
\renewcommand{\ge}{\geqslant}

\newcommand{\rg}[2]{\llbracket #1,#2  \rrbracket}
\newcommand{\card}[1]{\left | #1 \right |}
\newcommand{\sym}{\mathfrak{S}}
\newcommand{\scp}[2]{\left \langle #1 , #2  \right\rangle} %inner product

\newcommand{\F}{\mathbb{F}}
\newcommand{\G}{\mathbb{G}}
\newcommand{\K}{\mathbb{K}}

\newcommand{\Z}{\mathbb{Z}}

\newcommand{\fq}{\F_{q}}
\newcommand{\fqq}{\F_{q^2}}
\newcommand{\fqm}{\F_{q^m}}

\newcommand{\word}[1]{\vec{\boldsymbol{#1}}}
\newcommand{\av}{\word{a}}
\newcommand{\bv}{\word{b}}
\newcommand{\cv}{\word{c}}

\newcommand{\gv}{\word{g}}

\newcommand{\uv}{\word{u}}
\newcommand{\vv}{\word{v}}

\newcommand{\xv}{\word{x}}
\newcommand{\yv}{\word{y}}
\newcommand{\zv}{\word{z}}
\newcommand{\zerov}{\word{0}}
\newcommand{\onev}{\mathds{1}}
\newcommand{\tauv}{\word{\tau}}

\newcommand{\code}[1]{\mathscr{#1}}
\newcommand{\dual}[1]{{#1}^\bot}

\newcommand{\grs}{\mathsf{GRS}}
\newcommand{\rs}{\mathsf{RS}}

\newcommand{\AC}{\code{A}}
\newcommand{\BC}{\code{B}}
\newcommand{\CC}{\code{C}}
\newcommand{\DC}{\code{D}}
\newcommand{\EC}{\code{E}}
\newcommand{\NC}{\code{N}}
\newcommand{\sh}[2]{{\mathcal{S}}_{#1}\left(#2\right)}
\newcommand{\pu}[2]{{\mathcal{P}}_{#1}\left(#2\right)}

\newcommand{\mat}[1]{\boldsymbol{#1}}
\newcommand{\Am}{\mat{A}}
\newcommand{\Bm}{\mat{B}}
\newcommand{\Km}{\mat{K}}

\newcommand{\Hm}{\mat{H}}
\renewcommand{\Im}{\mat{I}}
\newcommand{\Sm}{\mat{S}}
\newcommand{\Tm}{\mat{T}}
\newcommand{\Um}{\mat{U}}
\newcommand{\Vm}{\mat{V}}

\DeclareMathOperator{\tr}{T}
\DeclareMathOperator{\nr}{N}

\newcommand{\inv}[2]{#1^{#2}}
\newcommand{\Hpub}{\Hm_{\mathsf{pub}}}

\newcommand{\Ginv}{\mat{G}_{\mathsf{inv}}}
\newcommand{\orbit}[1]{\F_2\cdot{}{#1}}
\newcommand{\orbitp}[1]{\left(\F_2\cdot{}{#1} \right)}

\def\nvarsU{n_U}
\def\nvarsX{n_V}
\def\nvarsB{n_B}
\def\nvarsT{n_T}

\def\k0{k_0}
\def\n0{n_0}
\def\a0{a_0}
\def\codim{c}

\begin{document}
\title{Practical Algebraic Attack on DAGS}

%
%\titlerunning{Abbreviated paper title}
% If the paper title is too long for the running head, you can set
% an abbreviated paper title here
%
\author{Magali Bardet\inst{1} \and
Manon Bertin\inst{1} \and
Alain Couvreur\inst{2} \and
Ayoub Otmani\inst{1}%%\orcidID{0000-0001-8176-8692}
}
%
%\authorrunning{M. Bardet et al.}
% First names are abbreviated in the running head.
% If there are more than two authors, 'et al.' is used.
%
\institute{LITIS, University of Rouen Normandie \\Avenue de l'universit\'e \\76801 Saint-\'Etienne-du-Rouvray, France 
\email{\{magali.bardet,manon.bertin8,Ayoub.Otmani\}@univ-rouen.fr}\\
\and
INRIA \& LIX, CNRS UMR 7161\\
\'Ecole polytechnique, 91128 Palaiseau Cedex, France\\
\email{alain.couvreur@lix.polytechnique.fr}}

\maketitle              
\begin{abstract}
  DAGS scheme is a key encapsulation mechanism (KEM) based on
  quasi-dyadic alternant codes that was submitted to NIST
  standardization process for a quantum resistant public key
  algorithm.  Recently an algebraic attack was devised by Barelli and
  Couvreur (Asiacrypt 2018) that efficiently recovers the private key.
  It shows that DAGS can be totally cryptanalysed by solving a system
  of bilinear polynomial equations.  However, some sets of DAGS
  parameters were not broken in practice.  In this paper we improve
  the algebraic attack by showing that the original approach was not
  optimal in terms of the ratio of the number of equations to the
  number of variables.  Contrary to the common belief that reducing at
  any cost the number of variables in a polynomial system is always
  beneficial, we actually observed that, provided that the ratio is
  increased and up to a threshold, the solving can be heavily improved
  by adding variables to the polynomial system. This enables us to
  recover the private keys in a few seconds. Furthermore, our
  experimentations also show that the maximum degree reached during
  the computation of the Gr\"obner basis is an important parameter
  that explains the efficiency of the attack. Finally, the authors of
  DAGS updated the parameters to take into account the algebraic
  cryptanalysis of Barelli and Couvreur. In the present article, we
  propose a hybrid approach that performs an exhaustive search on some
  variables and computes a Gr\"obner basis on the polynomial system
  involving the remaining variables. We then show that the updated set
  of parameters corresponding to 128-bit security can be broken with
  $2^{83}$ operations.

\end{abstract}

\section{Introduction}

The design of a quantum-safe public key encryption scheme is becoming an important issue with the recent process initiated by NIST to standardize one or more quantum-resistant public-key cryptographic algorithms.  One of the oldest cryptosystem that is not affected by the apparition of a large-scale quantum computer is the McEliece public key encryption scheme \cite{M78}.  It is a code-based cryptosystem that uses
the family of binary Goppa codes.
 The main advantage of this cryptosystem is its very fast encryption/decryption functions, and the fact that up to the present, nobody has succeeded to cryptanalyse it.

But in the eyes of those who are concerned with applications requiring very compact schemes, these positive aspects of the McEliece cryptosystem may not make up for its  large keys. For instance the classic McEliece \cite{BCLMNPPSSSW17} submitted to NIST uses 
at least 1MB public keys for 256 bits of security. A well-known approach for getting smaller keys consists in replacing binary Goppa codes by even more structured linear codes. A famous method started in \cite{G05} and further developed in \cite{BCGO09,MB09,BIGQUAKE,BBBCDGGHKNNPR17} relies on  codes displaying symmetries like cyclicity and dyadicity while having very efficient decoding algorithms.
Unlike the McEliece cryptosystem  which currently remains unbroken, the schemes 
\cite{BCGO09,MB09} are subject to efficient ciphertext-only attacks \cite{FOPT10} that recover the secret algebraic structure. The attack developed in \cite{FOPT10} formulates the general problem of recovering the algebraic structure of an alternant code as solving a system of polynomial equations. But it  involves very high degree polynomial equations with too many variables. Finding solutions to this kind of algebraic system is currently out of reach of the best known algorithms. However this approach turns out to be extremely fruitful when dealing with polynomial systems 
that come  from the quasi-cyclic \cite{BCGO09} and quasi-dyadic \cite{MB09} cryptosystems because the symmetries permit to reduce to a manageably small number of variables. 

The apparition of the algebraic attack  in \cite{FOPT10} generated a series of new algebraic attacks \cite{FOPPT14a,FOPPT15,FPP14} but since the original McEliece cryptosystem does not seem to be affected by this approach, it still raises the question of whether it represents a real threat.

Recently a new algebraic attack \cite{BC18} was mounted against DAGS \cite{BBBCDGGHKNNPR17} 
scheme. DAGS is a key encapsulation mechanism (KEM) based on quasi-dyadic alternant codes defined over quadratic extension. 
It was submitted to the standardisation process launched by NIST.
The attack relies on the component-wise product of codes in order to
build a system of bilinear multivariate equations.  The use of the
component-wise product of codes in cryptography is not new.  It first
appeared in \cite{W10} and has proved in several occasions
\cite{GOT12,GOT12a,CGGOT14,OT15} to be a powerful cryptanalytic tool
against algebraic codes like Generalised Reed-Solomon codes.  It even
enabled to mount for the first time a polynomial-time attack in
\cite{COT14,COT17} against a special family of non-binary Goppa codes
\cite{BLP10} displaying no symmetries.

\subsubsection{Our contribution.}

In this paper we improve the algebraic attack of \cite{BC18} by
showing that the original approach was not optimal in terms of the
ratio of the number of equations to the number of variables.  Contrary
to the common belief that reducing at any cost the number of variables
in a polynomial system is always beneficial, we actually observed
that, provided that the ratio is increased and up to a threshold, the
solving can be heavily improved by adding variables to the polynomial
system. This enables us to recover the private keys in a few
seconds. In Table~\ref{tab:comp} we report the average running times
of our attack and that of~\cite{BC18} performed on the same
machine. For DAGS-1 and DAGS-5, the linear algebra part of the attack
is the dominant cost.

Furthermore, our experimentations show that the maximum degree reached
during the computation of the Gr\"obner basis is an important
parameter that explains the efficiency of the attack. We observed that
the maximum degree never exceeds $4$.

Subsequently to the attack \cite{BC18}, the authors of DAGS updated
the parameters.  We
propose a hybrid approach that performs an exhaustive search on some
variables and computes a Gr\"obner basis on the polynomial system
involving the remaining variables. We then show that one set of
parameters does not have the claimed level of security.  Indeed the
parameters corresponding to 128-bit security can be broken with
$2^{83}$ operations.
 
\begin{table}[h]
    \begin{center}
        \caption{Running times  of the algebraic attack to break DAGS scheme. The computations are performed with Magma V2.23-1 on an Intel
Xeon processor clocked at 2.60GHz with 128Gb. We reproduced the computations from~\cite{BC18} on our machine.
  The columns ``Gr\"obner'' correspond to the Gr\"obner basis computation part, the columns ``Linear algebra'' to the linear algebra steps  of the attack.}
        \label{tab:comp}
        \begin{tabular}{@{}lcccc@{~~~}ccc@{}} \toprule
          Parameters & Security & \multicolumn{3}{c}{\cite{BC18}}  & \multicolumn{3}{c}{The present article}\\
                     && \ \ Gr\"obner \ \ & \ \ Linear \ \ & \ \ Total \ \ & \ \ Gr\"obner \ \  & \ \ Linear \ \  & \ \ Total \ \ \\
          && & algebra & & & algebra & \\
          \midrule
            DAGS-1 & $128$ & $552$s & $8$s &$560$s& $3.6$s & $6.4$s &$10$s \\
            DAGS-3 & $192$ & -- & -- & -- & $70$s & $16$s & $86$s \\
            DAGS-5 & $256$ & $6$s & $20$s  & $26$s & $0.5$s & $15.5$s& $16$s \\
            \bottomrule
        \end{tabular}
    \end{center}
\end{table}

\subsubsection{Organisation of the paper.} Section~\ref{sec:prelim} introduces the useful notation and important notions to describe the DAGS scheme. Section~\ref{sec:prod} recalls the important properties about the component-wise product of GRS and alternant codes. In Section~\ref{sec:attack} we describe the algebraic attack, and in Section~\ref{sec:result} the practical running times we obtained with our experimentations. Lastly, Section~\ref{sec:hybrid} explains the hybrid approach.

\section{Preliminaries} \label{sec:prelim}
\paragraph{\bf Notation.}
$\fq$ is the field with $q$ elements. In this paper, $q$ is a power of
$2$. For any $m$ and $n$ in $\Z$, $\rg{m}{n}$ is the set of integers
$i$ such that $m \leq i \leq n$.  The cardinality of set $A$ is
$\card{A}$.  Vectors and matrices are denoted by boldface letters as
$\av = (a_1,\dots{},a_n)$ and $\Am = (a_{i,j})$.  The \emph{Symmetric}
group on $n$ letters is denoted by $\sym_n$ and for any
$\vv = (v_1,\dots{},v_n)$ and $\sigma$ in $\sym_n$ we define
$\vv^\sigma \eqdef \left( v_{\sigma(1)},\dots{},v_{\sigma(n)}
\right)$.
The \emph{identity} matrix of size $n$ is written as $\Im_n$.
The transpose of a vector $\av$ and a matrix $\Am$ is denoted by
$\av^T$ and $\Am^T$.  The $i$-th row of a matrix $\Am = (a_{i,j})$ is
$\Am[i]$.
We recall that the \emph{Kronecker product} $\av\otimes \bv$ of two
vectors $\av = (a_1,\dots{},a_n)$ and $\bv$ is equal to
$(a_1\bv,\dots{},a_n \bv)$.  In particular, we denote by $\onev_n$
the all--one vector $\onev_n \eqdef (1,\dots, 1)\in \F_q^n$ and we have
$\onev_{n} \otimes \av = (\av,\dots{},\av)$ and
\[
\av \otimes  \onev_{n} = (a_1,\dots{},a_1,a_2,\dots{},a_2,\dots{},a_n,\dots{},a_n).
\] 

\medskip

Any $k$-dimensional vector subspace $\CC$ of $\F^n$ where $\F$ is field
is called a  \emph{linear code $\CC$ of length} $n$ and \emph{dimension} $k < n$ over $\F$. 
A matrix whose rows form a basis of $\CC$ is called a \emph{generator
  matrix}.  The \emph{orthogonal} or \emph{dual} of $\CC \subset \F^n$
is the linear space $\dual{\CC}$ containing all vectors $\zv$ from
$\F^n$ such that for all $\cv \in \CC$, we have
$\scp{\cv}{\zv} \eqdef \sum_{i=1}^n c_i z_i = 0$.  We always have
$\dim \dual{\CC} = n - \dim \CC$, and any generator matrix of
$\dual{\CC}$ is called a \emph{parity check} matrix of $U$.  The
\emph{punctured} code $\pu{I}{\CC}$ of $\CC$ over a set
$I \subset \rg{1}{n}$ is defined as
\[
\pu{I}{\CC} \eqdef  \Big\{ \uv \in \F^{n-\card{I}} ~\mid~  \exists \cv \in \CC, \; \; 
\uv = (c_i)_{i \in \rg{1}{n} \setminus I}  \Big\}.
\]
The \emph{shortened} $\sh{I}{\CC}$ code of $\CC$ over a set
$I \subset \rg{1}{n}$ is then defined as

\[
\sh{I}{\CC} \eqdef  \pu{I}{\Big\{ \cv \in \CC ~\mid ~
  \forall i \in I, \; \; c_i = 0     \Big\}}.
\]
We extend naturally the notation $\pu{I}{\cdot{}}$ to vectors and matrices.

\paragraph{\bf Algebraic codes.}

A \emph{generalized Reed-Solomon code} $\grs_t(\xv,\yv)$ of dimension $t$ and length $n$ where   $\xv$ is an $n$-tuple of distinct elements  from a finite field $\F$ 
 and  $\yv$ is an $n$-tuple  of non-zero elements from $\F$ is the linear code defined by
 \begin{equation}
\grs_{t}(\xv,\yv) 
 \eqdef  \Big \{ \big( y_1f(x_1), \dots{}, y_nf(x_n) \big)  ~\mid~ f \in \F_{< t }[z] \Big \}
 \end{equation}
 where $\F_{< t }[z]$ is the set of univariate polynomials $f$ with coefficients in $\F$ such that $\deg f < t$.
 The dimension of $\grs_{t}(\xv,\yv)$ is clearly $t$. 
 By convention $\grs_{t}(\xv,\onev_n)$ where $\onev_n$  is the all-one vector of length $n$ is simply a Reed-Solomon code denoted by $\rs_{t}(\xv)$.

The code $\dual{\grs_{t}}(\xv,\yv)$ is equal to $\grs_{n- t}(\xv,\dual{\yv})$
 where $\dual{\yv} = (\dual{y}_1,\dots{},\dual{y}_n)$ is the $n$-tuple  such that for all $j$ in $\rg{1}{n}$ 
 it holds that
 \begin{equation}   \label{dualgrs}
 \left(\dual{y}_j\right)^{-1} =y_j \prod_{\ell=1, \ell \neq j}^n (x_\ell - x_j).
 \end{equation}
An \emph{alternant code} $\AC_{t}(\xv,\yv)$ of degree $t \geq 1$  over a field $\K\subsetneq \F$ and length $n$
where $\xv$ is an $n$-tuple of distinct elements  from $\F^n$ 
 and  $\yv$ is an $n$-tuple  of non-zero elements from $\F^n$ is the linear code 
\begin{equation}
\AC_{t}(\xv,\yv) 
 \eqdef   \grs_{n-t}(\xv,\dual{\yv}) \cap \K^n = \dual{\grs}_{t}(\xv,\yv) \cap \K^n.
 \end{equation}
 The dimension of an alternant code satisfies the bound
$\dim  \AC_{t}(\xv,\yv)  \geq n - m t$ where $m$ is the degree of the field extension of $\F /\K$.
 \begin{remark}
Note that one has always the inclusions $\grs_{r}(\xv,\yv) \subseteq \grs_{t+r}(\xv,\yv)$ and $\AC_{t+r}(\xv,\yv) \subseteq \AC_{r}(\xv,\yv)$
for any  $r\geq 1$ and $t\geq 1$.
\end{remark}

\begin{proposition}\label{prop:affinetransform}
  Let $\grs_{k}(\xv,\yv)$ be a generalized Reed-Solomon code of
  dimension $k$ where $\xv$ is an $n$-tuple of distinct elements from
  $\fqm$ and $\yv$ is an $n$-tuple of non-zero elements from
  $\fqm$. For any affine map $\zeta : \fqm \rightarrow \fqm$ defined
  as $\zeta(z) \eqdef az +b$ where $a$ in $\fqm\setminus \{0\}$ and
  $b$ in $\fqm$ it holds that
\[
\grs_{k}\big(\zeta(\xv),\yv\big) =  \grs_{k}(\xv,\yv).
\]
\end{proposition}

\begin{remark}
  A consequence of Proposition~\ref{prop:affinetransform} is that it
  is possible to choose arbitrary values for two different coordinates
  $x_i$ and $x_j$ ($i \neq j$) provided that they are different.  For
  instance we may always assume that $x_1 = 0$ and $x_2 = 1$.

\smallskip

Another very important result from this proposition is that when an
affine map $\zeta(z) = az +b$ leaves globally invariant $\xv$ then it
induces a permutation $\sigma$ of $\sym_n$ thanks to the
identification:
 \[
 \forall i \in \rg{1}{n},  \; \; x_{\sigma(i)} \eqdef \zeta(x_i).  
  \]
  We call $\sigma$ the permutation \emph{induced} by the affine map
  $\zeta$. For the ease of notation we shall systematically identify
  $\sigma$ and $\zeta$.
\end{remark}

\paragraph{\bf Dyadic codes.}
A code
$\CC \subset \F^n$ is \emph{quasi-dyadic of order $2^\gamma$} where
$\gamma$ is a non-negative integer if there exists
$\G \subseteq \sym_n$ that is isomorphic to $\F_2^\gamma$ such that
\[
\forall (\sigma,\cv)  \in \G \times \CC, \;\; \cv^\sigma \in \CC.
\]

A construction of quasi-dyadic GRS and alternant codes is given in
\cite{FOPPT16}.  It considers $\gamma$ elements
$ b_1,\dots{},b_\gamma$ in $\fqm$ that are linearly independent over
$\F_2$.  The vector space $\oplus_{i=1}^\gamma \F_2 \cdot b_i$ generated
over $\F_2$ is then equal to a group $\G$. Next, it takes an $n_0$-tuple
$\tauv = (\tau_1\dots, \tau_{n_0})$ from $\fqm^{n_0}$ such that the
cosets $\tau_i + \G$ are pairwise disjoint, and finally it picks an
$n_0$-tuple $\yv = (y_1,\dots{},y_{n_0})$ composed of nonzero elements
from $\fqm$.
We now consider $\zv \eqdef \yv \otimes \onev_{2^\gamma}$ and
$\xv \eqdef \tauv \otimes \onev_{2^\gamma} + \onev_{n_0} \otimes \gv$
where $\gv \eqdef ( g )_{g \in \G}$.  The
action of $\G $ can then be described more explicitly: for any
$b$ in $\G$ we associate the translation defined for any $z $ in
$\fqm$ by $\sigma_b(z) \eqdef z + b$.  It is clear that $\sigma_b$
leaves globally invariant $\xv$ because we have
$\sigma_b (\G) = b + \G = \G$ and furthermore the following holds

\[
\sigma_b(\xv) = \xv
+ b \otimes \onev_{2^\gamma n_0} =
\tauv \otimes  \onev_{2^\gamma} +  \onev_{n_0} \otimes  \sigma_b(\gv).
\]

\begin{proposition}[\cite{FOPPT16}] \label{prop:QDGRS} Let $\K$ be a
  subfield of $\fqm$, $n_0$, $\gamma$, $\G$, $\yv$ and $\tauv$ defined
  as above and $n \eqdef 2^\gamma n_0$, $\gv \eqdef ( g )_{g \in \G}$,
  $\zv \eqdef \yv \otimes \onev_{2^\gamma}$ and
  $\xv \eqdef \tauv \otimes \onev_{2^\gamma} + \onev_{n_0} \otimes
  \gv$.  The codes $\grs_{r}(\xv,\zv) \subset \fqm^n$ and
  $\AC_{t}(\xv,\zv) \subset \K^n$ are quasi-dyadic of order
  $2^\gamma$.
 \end{proposition}

\begin{example}
  Let us take $n_0 = 2$ and $\gamma = 2$ then
  $\gv = (0, b_1, b_2, b_1 + b_2)$ and
  $\xv = ( \tau_1, \tau_1 + b_1, \tau_1 + b_2, \tau_1 + b_1 + b_2) || (
  \tau_2, \tau_2 + b_1, \tau_2 + b_2, \tau_2 + b_1 + b_2 )$. 
  The group $\G$ is then equal to $\{0,b_1,b_2,b_1+b_2\}$. We have $\sigma_{b_1}(\gv) = (b_1,0,b_2+b_1,b_2)$ and the
  permutation that corresponds to $b_1$ is  $(1 2) (3 4) (5 6) (7 8)$ in the
  canonical cycle notation.  \qed
\end{example}

\paragraph{\bf DAGS scheme.}
The public key encryption scheme DAGS \cite{BBB+18,webDAGS} submitted
to the NIST call for post quantum cryptographic proposals is a
McEliece-like scheme with a conversion to a KEM. It relies on
quasi-dyadic alternant codes
$\AC_{r} =\dual{\grs}_r(\xv,\zv) \cap \fq^n$ with $q = 2^s$ with
$\dual{\grs}_r(\xv,\zv) \subset \fqq^n$ ($m = 2$). The public code
$\AC_{r} $ is quasi-dyadic of order $2^\gamma$ where $\gamma \geq 1$.
The parameters are chosen such that $r = 2^\gamma r_0$ and
$n = 2^\gamma n_0$, and the dimension is $k = 2^\gamma k_0$ with
$k_0 \eqdef n_0 - 2 r_0$.

\medskip

Keeping up with the notation of Proposition~\ref{prop:QDGRS} the
vectors $\xv$ and $\zv$
can be written
as
$\xv \eqdef \tauv \otimes \onev_{2^\gamma} + \onev_{n_0} \otimes \gv$
and $\zv \eqdef \yv\otimes \onev_{2^\gamma}$ with
$\gv = ( g )_{g \in \G}$ where $\G = \oplus_{i=1}^{\gamma} \F_2 b_i$
is the vector space generated over $\F_2$ by $\gamma$ elements
$\bv = (b_1,\dots{},b_\gamma)$ that are linearly independent over
$\F_2$.  The quantities $\bv$ and $\yv$ are randomly drawn from
$\fqq^{n_0}$ such that the cosets $\tau_i + \G$ are pairwise disjoint
and $\yv$ is composed of nonzero elements in $\fqq$.

The public key is then an $(n-k)\times n$ parity check matrix $\Hpub$ of  $\AC_{r}$.
The quantities $(\bv, \tauv,\yv)$  have to be kept secret since 
they permit to decrypt any ciphertext. Table \ref{tab:dagsparam} gathers the parameters of the scheme. 
\begin{table}[h]
        \begin{center}
            \caption{DAGS-1, DAGS-3 and DAGS-5 correspond to the initial parameters (v1). When the algebraic attack \cite{BC18}  
appeared  the authors updated to DAGS-1.1, DAGS-3.1 and DAGS-5.1 (v2). }
            \label{tab:dagsparam}
            \begin{tabular}{@{}lcc*{5}{r}@{}} \toprule 
                Name & Security & $q$ & $m$ & $2^\gamma$ & $n_{0}$ & $k_{0}$ & $r_0 $ \\ 
                \midrule
                DAGS-1 & $128$ & $2^5$ & $2$ & $2^4$ & $52$  & $26$ & $13$ \\ 
                DAGS-3 & $192$ & $2^6$ & $2$ & $2^5$ & $38$ & $16$ & $11$ \\ 
                DAGS-5 & $256$ & $2^6$ & $2$ & $2^6$ & $33$ & $11$ & $11$ \\ \midrule
                DAGS-1.1 & $128$ & $2^6$ & $2$ & $2^4$ & $52$  & $26$ & $13$ \\
      		DAGS-3.1 & $192$ & $2^8$ & $2$ & $2^5$ & $38$ & $16$ & $11$ \\
	        DAGS-5.1 & $256$ & $2^8$ & $2$ & $2^5$ & $50$ & $28$ & $11$\\
                \bottomrule
            \end{tabular}
        \end{center}
    \end{table}

\section{Component-wise Product of Codes} \label{sec:prod}

An important property about GRS codes  is that
whenever $\av$ belongs to  $\grs_{r}(\xv,\yv)$  and 
$\bv$ belongs  to  $\grs_{t}(\xv,\zv)$, the component wise  product
$\av \star \bv   \eqdef (a_1b_1,\dots{},a_nb_n)$ belongs to $\grs_{t+r-1}(\xv,\yv\star \zv)$.  
Furthermore, if one defines the  component wise product  $\AC \star \BC$ of
 two linear codes $\AC \subset \F^n$ and $\BC \subset \F^n$  as the linear code
spanned by all the products $\av \star \bv$ with $\av$ in $\AC$ and $\bv$ in $\BC$ 
the inclusion is then an equality 
\begin{equation}
\grs_r(\xv,\yv) \star  \grs_t(\xv,\zv) = \grs_{r+t-1}(\xv,\yv \star \zv).
\end{equation}
In the case of alternant codes over a subfield $\K \subseteq \F$, one only gets in general the inclusion  
\begin{equation}
\left(\grs_{n-r}(\xv,\dual{\yv}) \cap \K^n \right) \star  \left(\grs_t(\xv,\onev_n) \cap \K^n\right) \subseteq \grs_{n-r+t-1}(\xv,\dual{\yv}) \cap \K^n 
\end{equation}
which leads to the following result.
\begin{proposition}[\cite{BC18}]  For any integer $r \geq 1$ and $t \geq 1$ the alternant codes 
 $\AC_{r+t-1}(\xv,\yv)$ and $\AC_{r}(\xv,\yv)$ over $\K \subsetneq \F$ where 
 $\xv$ is an $n$-tuple of distinct elements  from a finite field $\F$ 
 and  $\yv$ is an $n$-tuple  of non-zero elements from $\F$ satisfy the inclusion
\begin{equation}
\AC_{r+t-1}(\xv,\yv) \star  \left(\rs_t(\xv) \cap \K^n\right) \subseteq 
\AC_{r}(\xv,\yv). 
\end{equation}
\end{proposition}

The previous result is really interesting when $\rs_t(\xv) \cap \K^n$ is not the (trivial) code generated by 
$\onev_n$. 
This happens for instance  for $\F = \fqm$ and $\K = \fq$ when $t = q^{m-1} +1$ because 
$\rs_t(\xv) \cap \fq^n$ then always contains at least $\onev_n$ and  $\tr_{\fqm/\fq}(\xv) \eqdef  \left(\tr_{\fqm/\fq}(x_1),\dots{},\tr_{\fqm/\fq}(x_n)\right)$.
Actually one can observe that    $\tr_{\fqm/\fq}(\alpha \xv) $ also always belongs to
$\rs_t(\xv) \cap \fq^n$ for all $\alpha $ in $\fqm$.
Hence if $\{1,\omega_1,\dots{},\omega_{m-1}\}$ form an $\fq$-basis of $\fqm$ then 
$\alpha$ can be written as $\alpha_0 + \alpha_1 \omega_1 +\cdots{} + \alpha_{m-1} \omega_{m-1}$ with  $\alpha_0,\dots{},\alpha_{m-1}$ in $\fq$, and consequently with the convention that $\omega_0 \eqdef 1$
one has that
\[
\tr_{\fqm/\fq}(\alpha \xv)  =  \sum_{i=0}^{m-1}\alpha_i\tr_{\fqm/\fq}(\omega_i\xv).
\] 
This implies that $\dim_{\fq}  \rs_t(\xv) \cap \fq^n  \geq m+1$ when $t = q^{m-1}+1$.
Another interesting case is when $t = \frac{q^m-1}{q-1} +1$ 
then $\nr_{\fqm/\fq}(\xv) \eqdef  \left(\nr_{\fqm/\fq}(x_1),\dots{},\nr_{\fqm/\fq}(x_n) \right)$ 
belongs to $\rs_t(\xv) \cap \fq^n$, and one would get that
$\dim_{\fq}  \rs_t(\xv) \cap \fq^n  \geq m+2$.

\section{Algebraic Cryptanalysis}\label{sec:attack}

We present the ciphertext-only attack of \cite{BC18} that recovers the
private key of DAGS scheme. We refer to Section~\ref{sec:prelim} for
the notation. The public key is a parity-check matrix $\Hpub$ of a
quasi-dyadic alternant code $\AC_r$.  The attack recovers the secret
values $\bv = (b_1,\dots{},b_{\gamma})$,
$\tauv =(\tau_1,\dots{},\tau_{n_0})$ and
$ \yv = (y_1,\dots{},y_{n_0})$. The idea is to exploit the fact that
\begin{equation} \label{ARA} \AC_{r+t-1}(\xv,\yv\otimes
  \onev_{2^\gamma}) \star \left(\rs_{t}(\xv) \cap \fq^n\right)
  \subseteq \AC_{r}(\xv,\yv\otimes \onev_{2^\gamma}).
\end{equation}
where $\xv = \tauv \otimes \onev_{2^\gamma} + \onev_{n_0} \otimes \gv$
with $\gv = ( g )_{g \in \G}$ and
$\G = \oplus_{i=1}^{\gamma} \F_2 b_i$.
Because the secret vector $\yv$ is not anymore involved in the
definition of $\rs_{t}(\xv)$ an attacker gains a real advantage if she
manages to identify the codewords that are contained in
$\rs_{t}(\xv) \cap \fq^n$, especially when $t \geq q+2$ (see
\cite{BC18} for more details).
The attack of \cite{BC18} introduces the \emph{invariant code
  $\inv{\AC}{\G}_{r}$ with respect to $\G$ of $\AC_r$} which is
defined as
\[
  \inv{\AC}{\G}_{r} \eqdef \Big\{ (c_1,\dots{},c_n) \in \AC_r ~ \mid~
  \forall (i,j) \in \rg{0}{n_0-1} \times \rg{1}{2^\gamma},\;\; c_{i
    2^\gamma + j} = c_{i 2^\gamma + 1} \Big\}.
\]
The dimension of $\inv{\AC}{\G}_{r}$ is equal to $k_0$ (see
\cite{B17,BC18}). The cryptanalysis relies then on finding two vector
spaces $\DC$ and $\NC$ such that the constraints given in
\eqref{polsys} hold
\begin{equation} \label{polsys}
\left \{
\begin{array}{rcl}
\DC &\subsetneq & \inv{\AC}{\G}_{r}, \\
\dim \DC &=& k_0 -  c ~~\text{ where  } 
c \eqdef \frac{mq^{m-1}}{2^{\gamma}} = \frac{q}{2^{\gamma-1}},\\
\DC \star\NC & \subseteq  & \AC_r.
\end{array}
\right.
\end{equation}
Let us recall that $\AC_{r}$ and $\inv{\AC}{\G}_{r}$ are known,
especially it is simple to compute a generator matrix $\Ginv$ of
$\inv{\AC}{\G}_{r}$.  Since $\DC \subsetneq\inv{\AC}{\G}_{r}$ and
$\dim \DC = k_0 - c$ there exists a $(k_0 - c) \times k_0$ matrix
$\Km$ such that $\Km \Ginv$ generates $\DC$.  On the other hand $\NC$
necessarily satisfies the inclusion\footnote{It was observed
  experimentally in \cite{BC18} that actually the inclusion is most of
  the time an equality.}
$\NC \subseteq \dual{\Big(\DC \star \dual{\AC}_{r} \Big)}$ and
consequently $\big (\Km \Ginv \big) \star \Hpub $ is a parity check
matrix of $\NC$.  We are now able to state (without proof) an
important result justifying the interest of this approach.

\begin{theorem}[\cite{BC18}] \label{sol2polsys} Let us assume that
  $\card{\G} \leq q$.  Let $\DC$ be the invariant code of
  $\AC_{r+q}(\xv,\zv)$ and let $\NC$ be the vector space generated
  over $\fq$ by $\onev_n$, $\tr_{\fqq/\fq}(\xv) $
  $\tr_{\fqq/\fq}(\omega \xv) $ and $\nr_{\fqq/\fq}(\xv) $ where
  $\{1,\omega \}$ is an $\fq$-basis of $\fqq$.  Then $\DC$ and $\NC$
  are solution to \eqref{polsys}.
\end{theorem}

\begin{remark}
  Considering now the vector space generated by $\NC$ over $\fqq$ one
  can see that $\xv$ is also solution to \eqref{polsys} using this
  simple identity
\[
\xv = ( \omega^q - \omega) ^{-1} \left ( \omega^q \tr_{\fqq/\fq}(\xv) - \tr_{\fqq/\fq}(\omega \xv) \right).
\]
\end{remark}
The algebraic attack of \cite{BC18} recovers $\DC$ and $\NC$
satisfying \eqref{polsys} by introducing two sets of variables
$\Vm= (V_1,\dots{},V_n)$ and $\Km = (K_{i,j}) $ with
$i \in \rg{1}{k_0 - c}$ and $j \in \rg{1}{k_0}$ that satisfy the
multivariate quadratic system
 \[
 \big(\Km \Ginv \big) \star \Hpub  \cdot \Vm^T = \zerov.
 \]
 The number of variables of this system can be very high which is a
 hurdle to solving it efficiently in practice.  However this algebraic
 system does not take into account three observations that enable us
 to significantly reduce the number of variables.

\begin{itemize}
\item We know by Theorem~\ref{sol2polsys} that $\tr_{\fqq/\fq}(\xv) $
  (and $\xv$) are solution to \eqref{polsys} which means that we may
  assume that $\Vm$ has a ``quasi-dyadic'' structure.  We define two
  sets of variables $\Tm = (T_1,\dots,T_{n_0})$ and
  $\Bm = (B_1,\dots,B_\gamma)$ so that we can write
  $\Vm = \Tm \otimes \onev_{2^\gamma} + \onev_{n_0}\otimes
  \orbitp{\Bm}$ where
  $\orbit{\Bm} \eqdef \oplus_{i=1}^{\gamma} \F_2 B_i$ is the vector
formed by all the elements in the $\F_2$-vector space generated by
  $\Bm$. More precisely, $\orbit{(B_1)} = (0, B_1)$ and by induction,
  $\orbit{(B_1,\dots,B_i)} = \orbit{(B_1,\dots,B_{i-1})} ||
  \left(B_i\otimes \onev_{2^\gamma} +
    \orbit{(B_1,\dots,B_{i-1})}\right)$. For instance,

  $\orbit{(B_1, B_2, B_3)} = (0, B_1, B_2, B_1 + B_2, B_3, B_1 + B_3, B_2 + B_3, B_1 + B_2 + B_3)$.
  \medskip

\item Thanks to the shortening of $\DC$ and the puncturing of $\NC$ we
  are able to even more reduce the number of unknowns because for any
 $I \subset \rg{1}{n}$ it holds
\begin{equation} \label{shortalsy}
\sh{I}{\DC} \star \pu{I}{\NC} \subseteq \sh{I}{\AC_r}.
\end{equation}

\medskip

\item Lastly, if the first $(k_0 - c)$ columns of $\Km$ form an
  invertible matrix $\Sm$, we can then multiply the polynomial system
  by $\Sm^{-1}$ without altering the solution set. Therefore we may
  assume that the first columns of $\Km$ forms the identity matrix, i.e. $\Km = 	\begin{pmatrix} \Im_{d} & \Um  \end{pmatrix}$. Of
  course this observation also applies when considering the polynomial
  system defined in \eqref{shortalsy}.
 \end{itemize}

\medskip

The algebraic attack harnesses \eqref{shortalsy} by first picking a
set $I \subset \rg{1}{n}$ of cardinality $2^\gamma a_0$ such that $I$
is the union of $a_0$ disjoint dyadic blocks.  The different steps are
described below:
\begin{enumerate}
	
\item \label{step:solve} Recover $\sh{I}{\DC}$ and $\pu{I}{\NC}$ by
  solving the quadratic system
      \begin{equation} \label{eq:algebraicmodeling}
	\begin{pmatrix} \Im_{d} & \Um  \end{pmatrix} \sh{I}{\Ginv}  \star \pu{I}{\Hpub}  \cdot \pu{I}{\Vm}^T = \zerov
	\end{equation}
	where $d \eqdef \dim \sh{I}{\DC} = k_0 -c - a_0$ and
    $\Vm = \Tm \otimes \onev_{2^\gamma} + \onev_{n_0}\otimes
    \orbitp{\Bm}$.
	
	\medskip
	
  \item \label{step:pix} Reconstruct $\pu{I}{\xv}$ from $\sh{I}{\DC}$
    and $\pu{I}{\NC}$.
	
	\medskip
	
	\item \label{step:xy} Recover $\pu{I}{\yv}$ then $\xv$ and $\yv$ using linear algebra.
\end{enumerate}

\begin{remark}
Because of the particular form of the generator matrix
of $\sh{I}{\DC}$ in \eqref{eq:algebraicmodeling}, the system may have
only a trivial solution.  In other words, the polynomial system will
provide $\sh{I}{\DC}$ if it admits a generator matrix such that its
first $d$ columns form an invertible matrix, which holds with
probability $\prod_{i = 1}^{d} \left ( 1 - q^{-i} \right)$. 
\end{remark}

\begin{remark}\label{rem:dual}
  The attack searches for a code $\DC \subset \inv{\AC_r}{\G}$ of dimension $\k0-\codim$. When $\codim \ge \k0$, like DAGS-3.1 where $\k0 = c =16$, 
  this code does not exist. But it is possible to use $\dual{\AC_r}$ instead of $\AC_r$ to search for a code $\EC \subset \inv{(\dual{\AC_r})}{\G}$ of dimension $\n0-\k0-\codim$ such that
  \begin{equation}
    \label{eq:dual}    
 \EC \star \NC \subset \dual{\AC_r}
\end{equation}
We will not present this version here because it does not give practical improvements. 
\end{remark}

The authors in~\cite{BC18} addressed Steps~\ref{step:pix}
and~\ref{step:xy} by showing that it relies only on linear algebra
with matrices of size at most $n$ and can be done in $O(n^{3})$ operations.  
In this paper, we are mainly interested in Step~\ref{step:solve} and the
computation of a Gr\"obner basis of \eqref{eq:algebraicmodeling}.

\paragraph{\bf Solving \eqref{eq:algebraicmodeling} using Gr\"obner bases.}
Using Prop.~\ref{prop:affinetransform} and the fact that an affine map
preserves the quasi-dyadic structure, we can assume that $b_1=1$ and
$\tau_{n_0}=0$. Moreover, any vector in the code $\pu{I}{\NC}$ is a
solution to the system, in particular
$\tr_{\fqq/\fq}(b_2)^{-1}\tr_{\fqq/\fq}(\pu{I}{\xv})$, so we know that
a solution with $B_2=1$ exists.  

\begin{remark} \label{rem:inv_b2}
  Note that $\tr_{\fqq/\fq}(b_2)$ is not invertible if $b_2\in\fq$,
  which arises with probability $\frac{1}{q}$, but in this case a
  solution with $B_2=0, B_3=1$ exists in the system and we may
  specialize one more variable. 
\end{remark}

The following theorem identifies in \eqref{eq:algebraicmodeling} the number of equations and
variables.
\begin{theorem}\label{theo:numbereq}
If $\Ginv$ and $\Hpub$ are in systematic form and $I=\rg{1}{\a02^\gamma}$  then the polynomial system~\eqref{eq:algebraicmodeling} with $T_{\n0}=0, B_1=0, B_2=1$ contains
  \begin{equation}
  \begin{cases}
\nvarsU =  \left(\k0 - \codim - \a0\right) \codim & \text{ variables in } \Um\\
\nvarsT = \n0 -\k0 + \codim -1  &\text{ variables  }  T_{\k0-\codim+1},\dots,T_{\n0-1}\\
\nvarsB = \gamma - 2  &\text{ variables  } B_3,\dots, B_\gamma \\
\left(\k0 - \codim - \a0 \right)(\n0-\k0 -1) & \text{ quadratic equations}
\end{cases}\label{eq:nbvarsa0}
\end{equation}
that are bilinear in the variables $\Um$ and the variables in $\Vm$,
as well as $ \k0-\codim - \a0$ equations of the form
\[
T_i = P_i(\Um[i], T_{\n0-\k0+1},\dots,T_{\n0-1}, \Bm), \quad i \in \rg{1}{\k0-\codim-\a0}\]
where $P_i$ is a bilinear polynomial in the variables $\Um$ and $\Vm$.
\end{theorem}
\begin{proof}
   The hypothesis that $\Vm = \Tm \otimes \onev_{2^\gamma} + \onev_{n_0}\otimes
\orbitp{\Bm}$ reduces the number of variables, the number of solutions of the system (it restricts to quasi-dyadic solutions), but it also reduces the number of equations in the system by a factor $2^\gamma$. Indeed, if we consider two rows $\cv$ and $\cv^\sigma$ from $\Hpub$ in the same quasi-dyadic block, where $\sigma\in\G$, then for any row $\uv = \uv^\sigma$ from the invariant matrix $	\begin{pmatrix} \Im_{d} & \Um  \end{pmatrix} \sh{I}{\Ginv}$, the component-wise product with $\cv^\sigma$ satisfies $\uv\star\cv^\sigma = (\uv\star\cv)^\sigma$ and the resulting equations are $(\uv\star\cv)^\sigma \cdot {\Vm}^T = (\uv\star\cv \cdot {\Vm}^T)^\sigma = \uv\star\cv \cdot {\Vm}^T$ as
$\Vm = \Vm^\sigma + \sigma \onev_n$ and $\uv\star\cv\cdot\onev_n=0$.

The $i$-th row of $ \begin{pmatrix} \Im_{d} & \Um \end{pmatrix} \Ginv$
contributes to $\n0-\k0$ equations  that contain the
variables in the $i$-row $\Um[i]$ of $\Um$,  $\Bm$ and 
$T_{i}, T_{\k0-\codim+1},\dots,T_{\n0-1}$. Moreover, 
 by using the fact that the matrices are in systematic form,
the component-wise product of this row by the $i$-th row of $\Hpub$ (the
row in the $i$-th block as we take only one row every $2^\gamma$ rows) gives a
particular equation that expresses $T_i$ in terms of
$T_{\n0-\k0+1},\dots,T_{\n0-1}$, $\Bm$ and $\Um[i]$.  \qed
\end{proof}

\section{Experimental Results}\label{sec:result}

This section is devoted to the experimental results we obtained for 
computing a Gr\"obner  basis.  We consider two approaches for solving \eqref{eq:algebraicmodeling}.  The first one consists in solving the system 
without resorting to the shortening of codes ($I = \emptyset$).  The 
second one treats  the cases where we solve the system by shortening on $I$ with different cardinalities.

We report  the tests we have done using Magma \cite{BCP97} on a machine with a
Intel\textsuperscript{\textregistered}
Xeon\textsuperscript{\textregistered} 2.60GHz processor.
We indicate in the tables the number of clock cycles of the CPU, given by Magma
using the \texttt{ClockCycles()} function, as well as the time taken
on our machine.  

\paragraph{\bf Solving~\eqref{eq:algebraicmodeling} without shortening.}

The polynomial systems for the original DAGS parameters 
(DAGS-1, DAGS-3 and DAGS-5) are so
overdetermined that it is possible to compute directly a
Gr\"obner basis of \eqref{eq:algebraicmodeling} without shortening 
($a_0=0$ \textit{i.e} $I = \emptyset$).
Table~\ref{tab:statsnonshortened} gives the number of variables and
equations for the DAGS parameters.  One can see that there are at least 3 times as many equations as variables for the original parameters.

Moreover, the complexity of computing  Gr\"obner bases  is related to the  highest degree of the polynomials reached during the computation. A precise analysis has been done for generic overdetermined homogeneous systems (called ``semi-regular'' systems) in~\cite{BFS04,BFSY05} and for the particular case of generic bilinear systems in~\cite{FSS11}.
For generic overdetermined systems,
this degree decreases when the number of polynomials increases. 
For  DAGS-1, DAGS-3 and DAGS-5, the highest degree  is small (3 or 4).  Linear equations then appear at that degree which explains why we are able to solve the systems,
even if the number of variables is quite large (up to 119 variables for DAGS-1).

\begin{table}
    \begin{center}
        \caption{Computation time for  the Gr\"obner basis of \eqref{eq:algebraicmodeling} without shortening $\DC$ ($I = \emptyset$). 
        The columns ``$\dim \DC$'' and ``$c$'' correspond to the dimension of $\DC$ and the value of  $c = \frac{q}{2^{\gamma-1}}$. 
        The columns
          $\nvarsU$ and $\nvarsX = \nvarsB + \nvarsT$ give the number of variables $\Um$ 
    and $ \Bm \cup \Tm$ respectively.  The column ``Var.'' is equal to the total number of variables $\nvarsU+\nvarsX$ while the column ``Eq.'' indicates the number of
    equations in \eqref{eq:algebraicmodeling}. 
    ``Ratio'' is equal to the ratio of the number of equations to the number of variables. 
    ``Gr\"ob.'' gives the number of CPU clock cycles as given by the \texttt{ClockCycles()} function in Magma on Intel\textsuperscript{\textregistered}
Xeon\textsuperscript{\textregistered} 2.60GHz processor. 
    ``Deg.'' gives the degree where linear
    equations are produced. The column ``Mat. size'' is the size of the biggest matrix obtained during the computations.} 
        \label{tab:statsnonshortened}
        \begin{tabular}{@{}lccrrrrrccrcr@{}} 
        \toprule
            Param. & $\dim \DC$& $\codim$ & $\nvarsU$ & $\nvarsX$ & Var. & Eq. & Ratio & Gr\"ob.   & Deg. & \multicolumn{3}{c}{Mat. size}\\
          \midrule
            DAGS-1 & 22 & 4 & 88 & 31 & 119 & 550 & 4.6 &$2^{44}$ &  3 & 314,384 &$\times$& 401,540\\
            DAGS-3 & 12 & 4 & 48 & 28 & 76 & 252 & 3.3 &$2^{44}$ & 4 & 725,895 &$\times$& 671,071\\
            DAGS-5 & 9 & 2 & 18 & 27 & 45 & 189 & 4.2 & $2^{33}$ &  3 & 100,154 &$\times $& 8,019\\
            \bottomrule
        \end{tabular}
    \end{center}
  \end{table}

  \paragraph{\bf Solving~\eqref{eq:algebraicmodeling} with  shortening.}
  
  Shortening
  \eqref{eq:algebraicmodeling} on the $i$-th  dyadic block
  consists exactly in selecting a subset of the system that does not
  contain variables from $\Um[i]$. If we are able to shorten the
  system without increasing the degree where linear equations appear,
  then the Gr\"obner basis computation is faster because the matrices
  are smaller.

  For each set of parameters and for different dimensions of 
  $\sh{I}{\DC}$, we ran 100 tests. The results are shown in
  Table~\ref{tab:nbcalc}. For these tests we always assumed that
  $b_2\notin\fq$ (see Remark~\ref{rem:inv_b2} for more details).
  We recall that when $b_2 \in\fq$ we may specialize one more variable, that is to say we take $B_1 = 0$, $B_2=0$ and $B_3=1$, and we are able to solve the system.
  %%%
 In Table~\ref{tab:nbcalc} we can see that the best results  are obtained with
 $\dim \sh{I}{\DC} = 4$ for {DAGS-1}, 
 even if the number of variables is not the lowest.
 This can be explained for {DAGS-1} by the fact that highest degree is $3$  while when 
$\dim \sh{I}{\DC} = 2$ or $3$, the highest degree is $4$. 
 
 The figures obtained for  {DAGS-3} show that the highest degree is  $4$  
 for any  value of $\dim \sh{I}{\DC} \geq 3$. But when $\dim \sh{I}{\DC} = 4$
 more linear equations are produced because the ratio of the number of equations to the number of variables is larger.
 With a $\dim \sh{I}{\DC} = 2$, as Barelli and Couvreur~\cite{BC18} used
in their attack, the ratio  is small ($1.17$), and the maximal degree reached during the
Gr\"obner basis computation is too large ($\geq 6$) to get a reasonable
complexity. We had to stop because of a lack of memory and there was no linear equation at this degree.

 For DAGS-5  the value of  $\dim \sh{I}{\DC}$ has less influence on the performances because the system always has linear equations in degree $3$.
%%%
\begin{table}
    \begin{center}
        \caption{Computation time for  the Gr\"obner basis in \eqref{eq:algebraicmodeling} with the shortening of $\DC$. The column ``Gr\"ob.'' gives the number of CPU clock cycles as given by the \texttt{ClockCycles()} function in Magma on Intel\textsuperscript{\textregistered}
Xeon\textsuperscript{\textregistered} 2.60GHz processor.  
}
        \label{tab:nbcalc}
        \begin{tabular}{@{}ccccccrccrcr@{}} \toprule 
            Param. & $\dim\sh{I}{\DC}$ & Var. & Eq. & Ratio & Gr\"ob. & Time & Mem. (GB) & Deg. &  \multicolumn{3}{c}{Mat. size}\\ 
            \midrule
            DAGS-1 & 2 & 39 & 50 & 1.28 & $2^{39}$ & $276$s & $2.21$ & 4 & 76,392&$\times$& 62,518 \\
                  & 3 & 43 & 75 & 1.74 & $2^{38}$ & $163$s & $1.11$ & 4 & 97,908& $\times$& 87,238 \\ 
                  & 4 & 47 & 100 & 2.13 & $2^{33}$ & $4$s & $0.12$ & 3 & 11,487& $\times$ & 9,471 \\
                  & 5 & 51 & 125 & 2.45 & $2^{34}$ & $6$s & $0.24$ & 3 & 11,389& $\times$ & 13,805 \\
            \midrule
            DAGS-3 & 2 & 36 & 42 & 1.17 & -- & -- & $\ge 139$ & $\ge 6$ & --  & & \\
                  & 3 & 40 & 63 & 1.58 & $2^{39}$ & $321$s & $1.24$ & 4 & 85,981 &$\times$ & 101,482 \\
                  & 4 & 44 & 84 & 1.91 & $2^{37}$ & $70$s & $1.11$ & 4 & 103,973 &$\times$& 97,980 \\
                  & 5 & 48 & 105 & 2.19 & $2^{38}$ & $140$s & $1.48$ & 4 & 170,256 &$\times$ & 161,067 \\
            \midrule
            DAGS-5 & 2 & 31 & 42 & 1.35 & $2^{31}$ & $0.4$s & $0.12$ & 3 & 6,663 &$\times$ & 4,313 \\
                  & 3 & 33 & 63 & 1.91 & $2^{31}$ & $0.4$s & $0.13$ & 3 & 4,137&$\times$ & 4,066 \\
                  & 4 & 35 & 84 & 2.40 & $2^{31}$ & $0.5$s & $0.15$ & 3 & 5,843&$\times$ & 4,799 \\
                  & 5 & 37 & 105 & 2.84 & $2^{31}$ & $0.4$s & $0.19$ & 3 & 6,009&$\times$ & 4,839 \\
            \bottomrule
        \end{tabular}
      \end{center}
\end{table}

\paragraph{\bf Updated DAGS parameters.}
After the publication of the attack in \cite{BC18}, the authors of
DAGS proposed new parameters on their website \cite{webDAGS}. They are
given in Table~\ref{tab:dagsparam}. The computation of the Gr\"obner
basis is no longer possible with this new set of parameters.  For
DAGS-1.1, the number of variables is so high that the computation
involves a matrix in degree $4$ with about two millions of rows, and
we could not perform the computation.  For {DAGS-3.1}, as
$\codim = \k0$, the code $\DC$ does not exist. Even by considering the
dual of the public code as suggested by Remark~\ref{rem:dual}, it was
not feasible to solve the system.  This is due to the fact the system
is underdetermined: the ratio is $0.7$ as shown in
Table~\ref{tab:statsnonshortened}.  As for {DAGS-5.1}, the system has
too many variables and the ratio is too low.

 \begin{table}[h]
    \begin{center}
    \caption{Updated parameters. Columns are the same as Table~\ref{tab:statsnonshortened}}
        \begin{tabular}{@{}lccrrrrr@{}}  
        \toprule
            Param. & $\dim \DC$& $\codim$ & $\nvarsU$ & $\nvarsX$ & Var. & Eq. & Ratio \\
\midrule
          DAGS-1.1 & 18 & 8 & 144 & 35 & 179 & 450 & 2.5 \\
          DAGS-3.1 & 0 & 16 & -- & --& -- & 0 & 0 \\
          DAGS-3.1 dual & 6 & 16 & 96 & 34 & 130 & 90 & 0.7\\
          DAGS-5.1 & 12 & 16 & 192 & 40 & 232 & 252& 1.1 \\
            \bottomrule
        \end{tabular}
    \end{center}
  \end{table}

  However, as the systems are bilinear, a simple approach consisting
  in specializing a set of variables permits to get linear
  equations. For instance, with the new parameters for DAGS-1.1, and
  according to Theorem~\ref{theo:numbereq}, the algebraic system
  contains $\k0-\codim = \dim\DC = 18$ sets of $\n0-\k0-1=25$
  equations bilinear in $\codim=8$ variables $\Um$ and
  $\n0-\k0+\codim + \gamma-3 = 35$ variables $\Vm$. If we specialize
  the $\Um$ variables in a set of equations, we get $25$ linear
  equations in $35$ variables $\Vm$. As $q=2^6$ for DAGS-1.1,
  specializing $16$ variables $\Um$ gives a set of $50$ linear
  equations in $35$ variables that can be solved in at most
  $35^3 < 2^{15.39}$ finite field operations, and breaking DAGS-1.1
  requires to test all values of the $16$ variables $\Um$ in $\fq$,
  hence $2^{96}$ specializations, leading to an attack with
  $2^{111.39}$ operations, below the $128$-bit security claim.  This
  new set of parameters for DAGS-1.1 clearly does not take into
  account this point.

  The next section will show how to
  cryptanalyze efficiently the DAGS-1.1 parameters.

\section{Hybrid Approach on DAGS v2}\label{sec:hybrid}

We will show in this section that a hybrid
approach mixing exhaustive search and Gr\"obner basis provides an
estimated work factor of $2^{83}$ for  DAGS-1.1.

As the algebraic system is still highly overdetermined with a ratio of
2.5 between the number of variables and the number of equations, we
can afford to reduce the number of variables by shortening $\DC$ over
$\a0$ dyadic blocks while keeping a ratio large enough. For instance
if $\a0=\k0-\codim -2$ we have $\dim\sh{I}{\DC}=2$ and a system of
$50$ bilinear equations in $51$ variables.
On the other hand, specializing some
variables as in the hybrid approach from~\cite{BFP09} permits to
increase the ratio. For each value of the variables we compute a Gr\"obner basis of the specialized
system. When the Gr\"obner basis is $\langle 1\rangle$, it means that
the system has no solution, and it permits to ``cut branches'' of the
exhaustive search. Experimentally, computations are quite fast for the
wrong guesses, because the Gr\"obner basis computation stops immediately
when 1 is found in the ideal.

Moreover, if we specialize an entire row
of $\codim$ variables $\Um$, as the equations are
bilinear in $\Um$ and $\Vm$, then we get $\n0-\k0-\a0$ linear
equations in $\Vm$, which reduces the number of
variables. 
Table~\ref{tab:statshybrid} gives the complexity of the Gr\"obner
basis computation for {DAGS1.1}, for $\codim = 8$ variables $\Um$
specialized and different values of $\dim\sh{I}{\DC}$. 

\begin{table}[h]
    \begin{center}
      \caption{Experimental complexity of Gr\"obner basis computations with $\codim = 8$ specializations on $\Um$ for DAGS1.1. The system after specialization contains ``Var'' remaining variables. The column ``Linear'' (resp. ``Bilinear'') gives the number of linear (resp. bilinear) equations in the system.
        ``False'' contains the complexity of the
Gr\"obner basis computation for a wrong specialization, ``True'' for a correct one. The last column gives the global
complexity of the attack if we have to test all possible values of the
variables in $\fq = \F_{2^6}$, that is $2^{6\times 8}\times$ False + True.} 
        \label{tab:statshybrid}
        \begin{tabular}{@{}crccccc@{}} \toprule 
             $\dim\sh{I}{\DC}$ & Var &Linear & Bilinear & False  & ~True  & ~Total  \\ 
            \midrule
             2 & $43$ & $25$ & $25$ & $2^{35}$ & $2^{36}$ & $2^{83}$\\
             3 & $51$ & $25$ & $50$ & $2^{35}$ & $2^{36}$ & $2^{83}$\\ 
             4 & $59$ & $25$ & $75$ & $2^{38}$ & $2^{39}$ & $2^{86}$\\
             5 & $67$ & $25$ & $100$ & $2^{40}$ & $2^{40}$ & $2^{88}$\\
            \bottomrule
        \end{tabular}
      \end{center}
\end{table}

\section*{Acknowledgements}

This work has been supported by  the French ANR projects
MANTA (ANR-15-CE39-0013) and
CBCRYPT (ANR-17-CE39-0007). The authors are extremely grateful to \'Elise Barelli for kindly giving  her Magma code and for helpful discussions.

%%
%\begin{table}
%\caption{Table captions should be placed above the
%tables.}\label{tab1}
%\begin{tabular}{|l|l|l|}
%\hline
%Heading level &  Example & Font size and style\\
%\hline
%Title (centered) &  {\Large\bfseries Lecture Notes} & 14 point, bold\\
%1st-level heading &  {\large\bfseries 1 Introduction} & 12 point, bold\\
%2nd-level heading & {\bfseries 2.1 Printing Area} & 10 point, bold\\
%3rd-level heading & {\bfseries Run-in Heading in Bold.} Text follows & 10 point, bold\\
%4th-level heading & {\itshape Lowest Level Heading.} Text follows & 10 point, italic\\
%\hline
%\end{tabular}
%\end{table}

%\begin{figure}
%\includegraphics[width=\textwidth]{fig1.eps}
%\caption{A figure caption is always placed below the illustration.
%Please note that short captions are centered, while long ones are
%justified by the macro package automatically.} \label{fig1}
%\end{figure}
%
%\begin{theorem}
%This is a sample theorem. The run-in heading is set in bold, while
%the following text appears in italics. Definitions, lemmas,
%propositions, and corollaries are styled the same way.
%\end{theorem}
%%
% the environments 'definition', 'lemma', 'proposition', 'corollary',
% 'remark', and 'example' are defined in the LLNCS documentclass as well.
%
%\begin{proof}
%Proofs, examples, and remarks have the initial word in italics,
%while the following text appears in normal font.
%\end{proof}

%
% ---- Bibliography ----
%
% BibTeX users should specify bibliography style 'splncs04'.
% References will then be sorted and formatted in the correct style.
%


\begin{thebibliography}{10}
\providecommand{\url}[1]{\texttt{#1}}
\providecommand{\urlprefix}{URL }
\providecommand{\doi}[1]{https://doi.org/#1}

\bibitem{BBB+18}
Banegas, G., Barreto, P., Odilon~Boidje, B., Cayrel, P.L., Ndollane~Dione, G.,
  Gaj, K., Gueye, C.T., Haeussler, R., Klamti, J., Ndiaye, O., Tri~Nguyen, D.,
  Persichetti, E., Ricardini, J.: {DAGS}: Key encapsulation using dyadic gs
  codes. Journal of Mathematical Cryptology  \textbf{12}(4),  221--239 (09
  2018)

\bibitem{webDAGS}
Banegas, G., Barreto, P.S.L.M., Boidje, B.O., Cayrel, P.L., Dione, G.N., Gaj,
  K., Gueye, C.T., Haeussler, R., Klamti, J.B., N’diaye, O., Nguyen, D.T.,
  Persichetti, E., Ricardini, J.E.: {DAGS}: Key encapsulation for dyadic {GS}
  codes, specifications v2 (09 2018)

\bibitem{BBBCDGGHKNNPR17}
Banegas, G., Barreto, P.S., Boidje, B.O., Cayrel, P.L., Dione, G.N., Gaj, K.,
  Gueye, C.T., Haeussler, R., Klamti, J.B., N'diaye, O., Nguyen, D.T.,
  Persichetti, E., Ricardini, J.E.: {D}{A}{G}{S} : Key encapsulation for dyadic
  {G}{S} codes.
  \url{https://csrc.nist.gov/CSRC/media/Projects/Post-Quantum-Cryptography/documents/round-1/submissions/DAGS.zip}
  (Nov 2017), first round submission to the NIST post-quantum cryptography call

\bibitem{BFS04}
Bardet, M., Faug{\`e}re, J.C., Salvy, B.: On the complexity of gr{\"o}bner
  basis computation of semi-regular overdetermined algebraic equations. In:
  ICPSS'04. pp. pp. 71--75. (2004), international Conference on Polynomial
  System Solving, November 24 - 25 - 26, Paris, France

\bibitem{BFSY05}
Bardet, M., Faug{\`e}re, J.C., Salvy, B., Yang, B.Y.: Asymptotic behaviour of
  the degree of regularity of semi-regular quadratic polynomial systems. In:
  MEGA'05. p. 15 p. (2005), eighth International Symposium on Effective Methods
  in Algebraic Geometry, Porto Conte, Alghero, Sardinia (Italy), May 27th --
  June 1st

\bibitem{BIGQUAKE}
Bardet, M., Barelli, {\'{E}}., Blazy, O., Cando~Torres, R., Couvreur, A.,
  Gaborit, P., Otmani, A., Sendrier, N., Tillich, J.P.: {BIGQUAKE}.
  \url{https://bigquake.inria.fr} (Nov 2017), {NIST} Round 1 submission for
  Post-Quantum Cryptography

\bibitem{B17}
Barelli, {\`{E}}.: {On the security of some compact keys for McEliece scheme}.
  In: WCC Workshop on Coding and Cryptography (Sep 2017)

\bibitem{BC18}
Barelli, {\'{E}}., Couvreur, A.: An efficient structural attack on {NIST}
  submission {DAGS}. In: Peyrin, T., Galbraith, S. (eds.) Advances in
  Cryptology -- ASIACRYPT 2018. pp. 93--118. Springer International Publishing,
  Cham (2018)

\bibitem{BCGO09}
Berger, T.P., Cayrel, P.L., Gaborit, P., Otmani, A.: Reducing key length of the
  {McEliece} cryptosystem. In: Preneel, B. (ed.) Advances in Cryptology -
  AFRICACRYPT~2009. LNCS, vol.~5580, pp. 77--97. Gammarth, Tunisia (Jun~21-25
  2009)

\bibitem{BCLMNPPSSSW17}
Bernstein, D.J., Chou, T., Lange, T., von Maurich, I., Niederhagen, R.,
  Persichetti, E., Peters, C., Schwabe, P., Sendrier, N., Szefer, J., Wen, W.:
  Classic {M}c{E}liece: conservative code-based cryptography (Nov 2017), first
  round submission to the NIST post-quantum cryptography call

\bibitem{BLP10}
Bernstein, D.J., Lange, T., Peters, C.: Wild {M}c{E}liece. In: Biryukov, A.,
  Gong, G., Stinson, D.R. (eds.) Selected Areas in Cryptography. LNCS,
  vol.~6544, pp. 143--158 (2010)

\bibitem{BFP09}
Bettale, L., Faug\`ere, J.C., Perret, L.: Hybrid approach for solving
  multivariate systems over finite fields. Journal of Mathematical Cryptology
  \textbf{3}(3),  177--197 (2009)

\bibitem{BCP97}
Bosma, W., Cannon, J., Playoust, C.: The {Magma} algebra system {I}: The user
  language. J. Symbolic Comput.  \textbf{24}(3/4),  235--265 (1997)

\bibitem{CGGOT14}
Couvreur, A., Gaborit, P., Gauthier{-}Uma{\~{n}}a, V., Otmani, A., Tillich,
  J.P.: Distinguisher-based attacks on public-key cryptosystems using
  {Reed-Solomon} codes. Des. Codes Cryptogr.  \textbf{73}(2),  641--666 (2014)

\bibitem{COT14}
Couvreur, A., Otmani, A., Tillich, J.P.: Polynomial time attack on wild
  {M}c{E}liece over quadratic extensions. In: Nguyen, P.Q., Oswald, E. (eds.)
  Advances in Cryptology - EUROCRYPT~2014. LNCS, vol.~8441, pp. 17--39.
  Springer Berlin Heidelberg (2014)

\bibitem{COT17}
Couvreur, A., Otmani, A., Tillich, J.P.: Polynomial time attack on wild
  {M}c{E}liece over quadratic extensions. IEEE Trans. Inform. Theory
  \textbf{63}(1),  404--427 (Jan 2017)

\bibitem{FOPPT14a}
Faug{\`{e}}re, J.C., Otmani, A., Perret, L., de~Portzamparc, F., Tillich, J.P.:
  Structural weakness of compact variants of the {McEliece} cryptosystem. In:
  Proc. IEEE Int. Symposium Inf. Theory - ISIT~2014. pp. 1717--1721. Honolulu,
  HI, USA (Jul 2014)

\bibitem{FOPPT16}
Faug{\`{e}}re, J.C., Otmani, A., Perret, L., de~Portzamparc, F., Tillich, J.P.:
  Folding alternant and {Goppa Codes} with non-trivial automorphism groups.
  IEEE Trans. Inform. Theory  \textbf{62}(1),  184--198 (2016)

\bibitem{FOPPT15}
Faug{\`{e}}re, J.C., Otmani, A., Perret, L., de~Portzamparc, F., Tillich, J.P.:
  Structural cryptanalysis of {McEliece} schemes with compact keys. Des. Codes
  Cryptogr.  \textbf{79}(1),  87--112 (2016)

\bibitem{FOPT10}
Faug{\`e}re, J.C., Otmani, A., Perret, L., Tillich, J.P.: Algebraic
  cryptanalysis of {McEliece} variants with compact keys. In: Advances in
  Cryptology - EUROCRYPT~2010. LNCS, vol.~6110, pp. 279--298 (2010)

\bibitem{FPP14}
Faug{\`{e}}re, J.C., Perret, L., de~Portzamparc, F.: Algebraic attack against
  variants of {McEliece} with {Goppa} polynomial of a special form. In:
  Advances in Cryptology - ASIACRYPT~2014. LNCS, vol.~8873, pp. 21--41.
  Springer, Kaoshiung, Taiwan, R.O.C. (Dec 2014)

\bibitem{FSS11}
{Faug{\`e}re}, J.C., {Safey El Din}, M., {Spaenlehauer}, P.J.: {Gr\"obner}
  bases of bihomogeneous ideals generated by polynomials of bidegree (1,1):
  Algorithms and complexity. J. Symbolic Comput.  \textbf{46}(4),  406--437
  (2011)

\bibitem{G05}
Gaborit, P.: Shorter keys for code based cryptography. In: Proceedings of the
  2005 International Workshop on Coding and Cryptography ({WCC} 2005). pp.
  81--91. Bergen, Norway (Mar 2005)

\bibitem{GOT12}
Gauthier, V., Otmani, A., Tillich, J.P.: A distinguisher-based attack of a
  homomorphic encryption scheme relying on {Reed-Solomon} codes. CoRR
  \textbf{abs/1203.6686} (2012)

\bibitem{GOT12a}
Gauthier, V., Otmani, A., Tillich, J.P.: A distinguisher-based attack on a
  variant of {McEliece's} cryptosystem based on {Reed-Solomon} codes. CoRR
  \textbf{abs/1204.6459} (2012)

\bibitem{M78}
McEliece, R.J.: A Public-Key System Based on Algebraic Coding Theory, pp.
  114--116. Jet Propulsion Lab (1978), dSN Progress Report 44

\bibitem{MB09}
Misoczki, R., Barreto, P.: Compact {McEliece} keys from {Goppa} codes. In:
  Selected Areas in Cryptography. Calgary, Canada (Aug~13-14 2009)

\bibitem{OT15}
Otmani, A., Tal{\'{e}}-Kalachi, H.: Square code attack on a modified
  {S}idelnikov cryptosystem. In: Hajji, S.E., Nitaj, A., Carlet, C., Souidi,
  E.M. (eds.) Codes, Cryptology, and Information Security - First International
  Conference, {C2SI} 2015, Rabat, Morocco, May 26-28, 2015, Proceedings - In
  Honor of Thierry Berger. Lecture Notes in Computer Science, vol.~9084, pp.
  173--183. Springer (2015)

\bibitem{W10}
Wieschebrink, C.: Cryptanalysis of the {Niederreiter} public key scheme based
  on {GRS} subcodes. In: Post-Quantum Cryptography~2010. LNCS, vol.~6061, pp.
  61--72. Springer (2010)

\end{thebibliography}
\end{document}